\documentclass[10pt]{llncs}

\usepackage{todo}
\usepackage{tipa}
\usepackage{algorithm}
\usepackage[noend]{algorithmic}

\usepackage{ltl}
\usepackage{amsmath}
\usepackage{mathtools}
\usepackage{stmaryrd}
\usepackage{mathrsfs}
\usepackage{xspace}
\DeclareMathAlphabet{\mathantt}{OT1}{antt}{li}{it}
\DeclareMathAlphabet{\mathpzc}{OT1}{pzc}{m}{it}
\usepackage[a]{esvect}
\usepackage{multirow}
\makeatletter
\def\moverlay{\mathpalette\mov@rlay}
\def\mov@rlay#1#2{\leavevmode\vtop{%
   \baselineskip\z@skip \lineskiplimit-\maxdimen
   \ialign{\hfil$\m@th#1##$\hfil\cr#2\crcr}}}
\newcommand{\charfusion}[3][\mathord]{
    #1{\ifx#1\mathop\vphantom{#2}\fi
        \mathpalette\mov@rlay{#2\cr#3}
      }
    \ifx#1\mathop\expandafter\displaylimits\fi}
\makeatother
\newcommand{\cupdot}{\charfusion[\mathbin]{\cup}{\cdot}}
\newcommand{\AP}{\textit{AP}}
\newcommand{\tr}{\textit{tr}}
\newcommand{\TR}{\textit{TR}}
\newcommand{\CQuan}{\#}

\newcommand{\hyperltl}{{\sc HyperLTL}\@\xspace}
\newcommand{\nats}{\mathbb N}

\newcommand{\pspace}{\textsc{Pspace}}

\usetikzlibrary{arrows,automata}
\usetikzlibrary{positioning}

\usepackage[firstpage]{draftwatermark}
\SetWatermarkText{\hspace*{3.5in}
\raisebox{6in}{\includegraphics[scale=0.1]{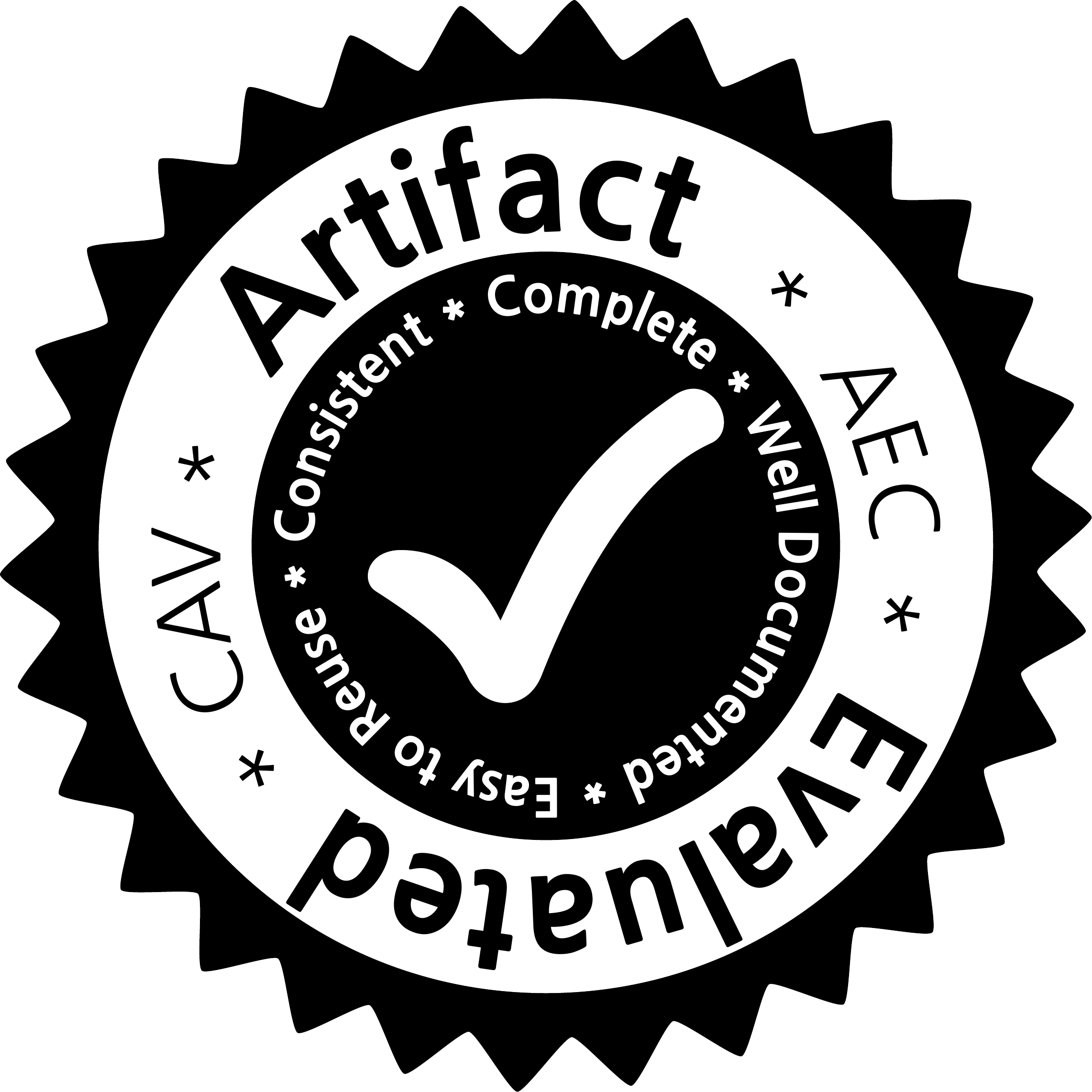}}}
\SetWatermarkAngle{0}
\begin{document}

\title{Model Checking Quantitative Hyperproperties\thanks{This work was partly supported by the ERC Grant 683300 (OSARES) and by the German Research Foundation (DFG) in the Collaborative Research Center 1223.}}
\author{Bernd Finkbeiner \and Christopher Hahn \and\\ Hazem Torfah}
\institute{Reactive Systems Group\\Saarland University\\\email{lastname@react.uni-saarland.de}}
\date{}
\maketitle

\begin{abstract}
  Hyperproperties are properties of sets of computation
  traces. In this paper, we study quantitative hyperproperties, which we
  define as hyperproperties 
  that express a bound on the number of traces that may appear
  in a certain relation. For example, quantitative non-interference
  limits the amount of information about certain secret inputs
  that is leaked through the observable outputs of a system.
  Quantitative non-interference thus bounds the
  number of traces that have the same observable input but different
  observable output.  
  We study quantitative hyperproperties in the
  setting of HyperLTL, a temporal logic for hyperproperties. We show
  that, while quantitative hyperproperties can be expressed in
  HyperLTL, the running time of the HyperLTL model checking algorithm
  is, depending on the type of property, exponential or even doubly exponential in the quantitative bound.  We improve this
  complexity with a new model checking algorithm based on model-counting. The new algorithm  needs only logarithmic space in the bound and therefore improves, depending on the property, exponentially or even doubly exponentially over the model checking algorithm of HyperLTL. In the worst case, the new algorithm needs polynomial space in the size of the system. Our Max\#Sat-based prototype implementation demonstrates, however, that the counting approach is viable on systems with nontrivial quantitative information flow requirements such as a passcode checker. 
\end{abstract}

\section{Introduction}

Model checking algorithms~\cite{Clarke/Design-and-Synthesis} are the cornerstone of computer-aided verification. As their input consists of both the system under verification and a logical formula that describes the property to be verified, they uniformly solve a wide range of verification problems, such as all verification problems expressible in linear-time temporal logic (LTL), computation-tree logic (CTL), or the modal $\mu$-calculus. Recently, there has been a lot of interest in extending  model checking from standard trace and tree properties to \emph{information flow} policies like observational determinism or quantitative information flow. Such policies are called \emph{hyperproperties}~\cite{journals/jcs/ClarksonS10} and can be expressed in HyperLTL~\cite{conf/post/ClarksonFKMRS14}, an extension of LTL with trace quantifiers and trace variables. For example, \emph{observational determinism}~\cite{conf/csfw/ZdancewicM03}, the requirement that any pair of traces that have the same observable input also have the same observable output, can be expressed as the following HyperLTL formula:
$ \forall \pi.\forall \pi'.\,(\LTLsquare \pi =_I {\pi'}) \rightarrow (\LTLsquare \pi =_O {\pi'}) $.
For many information flow policies of interest, including observational determinism, there is no longer a need for property-specific algorithms: it has been shown that the standard HyperLTL model checking algorithm~\cite{conf/cav/FinkbeinerRS15} performs just as well as a specialized algorithm for the respective property.

The class of hyperproperties studied in this paper is one where, by contrast, the standard model checking algorithm performes badly. We are interested in \emph{quantitative hyperproperties}, i.e., hyperproperties  that express a bound on the number of traces that may appear
  in a certain relation.
A prominent example of this class of properties is \emph{quantitative non-interference}~\cite{Smith/2009/OnTheFoundationsOfQantitativeInformationFlow,Yasuoka+Terauchi/2010/OnBoundingProblemsOfQuantitativeInformationFlow}, where we allow some flow of information but, at the same time, limit the amount of information that may be leaked. Such properties are used, for example, to describe the correct behavior of a password check, where some information flow is unavoidable (``the password was incorrect''), and perhaps some extra information flow is accceptable (``the password must contain a special character''), but the information should not suffice to guess the actual password.
In HyperLTL, quantitative non-interference can be expressed~\cite{conf/post/ClarksonFKMRS14} as the formula
	$
	%\mathit{QN}(c)\coloneqq
        \forall \pi_0.\ \forall \pi_1 \ldots \forall \pi_{2^c}.~  \left(\bigwedge_i \LTLsquare ({\pi_i} =_I {\pi_0}) \right) \rightarrow \left(\bigvee_{i \neq j} \LTLsquare ({\pi_i} =_O {\pi_j})\right).
	$
The formula states that there do not exist $2^c+1$ traces (corresponding to more than $c$ bits of information) with the same observable input but different observable output. The bad performance of the
standard model checking algorithm is a consequence of the fact that the $2^{c}+1$ traces are tracked simultaneously. For this purpose, the model checking algorithm builds and analyzes a $(2^c+1)$-fold self-composition of the system.

We present a new model checking algorithm for quantitative hyperproperties that avoids the construction of the huge self-composition. The key idea of our approach is to use \emph{counting} rather than \emph{checking} as the basic operation. Instead of building the self-composition and then \emph{checking} the satisfaction of the formula, we add new atomic propositions and then \emph{count} the number of sequences of evaluations of the new atomic propositions that satisfy the specification. Quantitative hyperproperties are expressions of the following form:
\[ \forall \pi_1. \dots \forall \pi_k.\, \varphi \rightarrow (\# \sigma: X.\, \psi \triangleleft n),\]
where $\triangleleft \in \{\leq,<,\geq,>,=\}$. The universal quantifiers introduce a set of reference traces against which other traces can be compared. The formulas $\varphi$ and $\psi$ are HyperLTL formulas. The counting quantifer $\# \sigma:X.\, \psi$ counts the number of paths $\sigma$ with different valuations of the atomic propositions $X$ that satisfy $\psi$.
The requirement that no more than $c$ bits of information are leaked is the following quantitative hyperproperty:
\[
\forall \pi.\, \CQuan \sigma\colon O. \,  \LTLsquare({\pi} =_I \sigma)  \leq 2^c
\]
As we show in the paper, such expressions do not change the expressiveness of the logic; however, they allow us to express quantitative hyperproperties in exponentially more concise form. The counting-based model checking algorithm then maintains this advantage with a logarithmic counter, resulting in exponentially better performance in both time and space.
%In addition, we provide a practical algorithm via encoding the

The viability of our counting-based model checking algorithm is demonstrated on a Max\#SAT-based prototype implementation. For quantitative hyperproperties of intrest, such as bounded leakage of a password checker, our algorithm shows promising results, as it significantly outperforms existing model checking approaches.

\subsection{Related Work}
Quantitative information-flow has been studied extensively in the literature. See, for example, the following selection of contributions on this topic:~\cite{Smith/2009/OnTheFoundationsOfQantitativeInformationFlow,DBLP:conf/ccs/KopfB07,DBLP:conf/csfw/ClarksonMS05,DBLP:journals/logcom/ClarkHM05,DBLP:journals/jcs/AlvimAP12,DBLP:conf/sp/Gray91}.
Multiple verification methods for quantitative information-flow were proposed for sequential systems. For example, with static analysis techniques~\cite{DBLP:journals/jcs/ClarkHM07}, approximation methods~\cite{DBLP:conf/csfw/KopfR10}, equivalence relations~\cite{cohen1978information,DBLP:conf/sp/BackesKR09}, and randomized methods~\cite{DBLP:conf/csfw/KopfR10}.
%There is a large number of existing approaches to quantitative information-flow control.
Quantitative information-flow for multi-threaded programs was considered in~\cite{DBLP:conf/pldi/ChenM07}.
%Measuring maximum flow via dynamic analysis~\cite{DBLP:conf/pldi/McCamantE08}.
%Using equivalence relations to analyze the information-flow of a system was proposed in~\ref{}

The study of quantitative information-flow in a reactive setting gained a lot of attention recently after the introduction of hyperproperties~\cite{journals/jcs/ClarksonS10} and the idea of verifying the self-composition of a reactive system~\cite{DBLP:journals/mscs/BartheDR11} in order to relate traces to each other. There are several possibilities to measure the amount of leakage, such as Shannon entropy~\cite{DBLP:books/aw/Denning82,DBLP:journals/jcs/ClarkHM07,DBLP:conf/popl/Malacaria07}, guessing entropy~\cite{DBLP:conf/ccs/KopfB07,DBLP:conf/sp/BackesKR09}, and min-entropy~\cite{Smith/2009/OnTheFoundationsOfQantitativeInformationFlow}.
%In this paper, we consider the latter, which quantifies the worst-case amount of information an attacker can gain given the answer to a single guess about the secret.
A classification of quantitative information-flow policies as safety and liveness hyperproperties was given in~\cite{DBLP:journals/tcs/YasuokaT14}.
While several verification techniques for hyperproperties exists~\cite{DBLP:journals/jfp/BanerjeeN05,DBLP:journals/ijisec/HammerS09,DBLP:conf/nordsec/MilushevC13,DBLP:conf/popl/Myers99}, the literature was missing general approaches to quantitative information-flow control. SecLTL~\cite{DBLP:conf/vmcai/DimitrovaFKRS12} was introduced as first general approach to model check (quantitative) hyperproperties, before HyperLTL~\cite{conf/post/ClarksonFKMRS14}, and its corresponding model checker~\cite{conf/cav/FinkbeinerRS15}, was introduced as a temporal logic for hyperproperties, which subsumes the previous approaches.

Using counting to compute the number of solutions of a given formula is studied in the literature as well and includes many probabilistic inference problems, such as Bayesian net reasoning~\cite{DBLP:journals/jar/LittmanMP01}, and planning problems, such as computing robustness of plans in incomplete domains~\cite{DBLP:conf/aaai/MorwoodB12}. State-of-the-art tools for propositional model counting are \texttt{Relsat}~\cite{DBLP:conf/aaai/BayardoS97} and \texttt{c2d}~\cite{DBLP:conf/ecai/Darwiche04}. Algorithms for counting models of temporal logics and automata over infinite words have been introduced in \cite{DBLP:conf/lata/FinkbeinerT14,FTATVA17,Torfah2016}. 
The counting of projected models, i.e., when some parts of the models are irrelevant, was studied in~\cite{DBLP:conf/sat/AzizCMS15}, for which tools such as \texttt{\#CLASP}~\cite{DBLP:conf/sat/AzizCMS15} and~\texttt{DSharp\_P}~\cite{DBLP:conf/sat/AzizCMS15,DBLP:conf/ai/MuiseMBH12} exist. Our SAT-based prototype implementation is based on a reduction to a Max\#SAT~\cite{DBLP:conf/aaai/FremontRS17} instance, for which a corresponding tool exists.

Among the already existing tools for computing the amount of information leakage, for example, \texttt{QUAIL}~\cite{DBLP:conf/cav/BiondiLTW13}, which analyzes programs written in a specific while-language and \texttt{LeakWatch}~\cite{DBLP:conf/esorics/ChothiaKN14}, which estimates the amount of leakage in Java programs, \texttt{Moped-QLeak}~\cite{DBLP:conf/fsttcs/ChadhaMS14} is closest to our approach. However, their approach of computing a symbolic summary as an Algebraic Decision Diagram is, in contrast to our approach, solely based on model counting, not maximum model counting.
%\todo{Max#sat, projected model counting, approximative projected model counting for sat}
%~\todo{mopedqleak}

\section{Preliminaries}
\subsection{HyperLTL}

HyperLTL~\cite{conf/post/ClarksonFKMRS14} extends linear-time temporal logic (LTL) with trace variables and trace quantifiers.
\label{tracevshyper}
%We model a reactive system as a set of its possible infinite execution traces.
Let $\mathit{AP}$ be a set of \emph{atomic propositions}.
A \emph{trace} $t$ is an infinite sequence over subsets of the atomic propositions. We define the set of traces $\mathit{TR} \coloneqq (2^\mathit{AP})^\omega$.
A subset $T \subseteq \mathit{TR}$ is called a \emph{trace property} and a subset $H \subseteq 2^\mathit{TR}$ is called a \emph{hyperproperty}.
We use the following notation to manipulate traces:
let $t \in \mathit{TR}$ be a trace and $i \in \mathbb{N}$ be a natural number. $t[i]$ denotes the $i$-th element of $t$. Therefore, $t[0]$ represents the starting element of the trace. Let $j \in \mathbb{N}$ and $j \geq i$. $t[i,j]$ denotes the sequence $t[i]~t[i+1]\ldots t[j-1]~t[j]$. $t[i, \infty]$ denotes the infinite suffix of $t$ starting at position $i$.

\paragraph{HyperLTL Syntax.}
Let $\mathcal{V}$ be an infinite supply of trace variables.  The syntax of HyperLTL is given by the following grammar:
\begin{align*}
	\psi~&\Coloneqq~\exists \pi.\;\psi~~|~~\forall\pi.\;\psi~~|~~\varphi \\
	\varphi~&\Coloneqq~a_{\pi}~~|~~\neg \varphi~~|~~\varphi \vee \varphi~~|~~\LTLnext \varphi~~|~~\varphi\, \LTLuntil \varphi 
\end{align*}
where $a\in \mathit{AP}$ is an atomic proposition and $\pi \in \mathcal V$ is a trace variable. Note that atomic propositions are indexed by trace variables.
The quantification over traces makes it possible to express properties like ``on all traces $\psi$ must hold'', which is expressed by $\forall \pi.~\psi$. Dually, one can express that ``there exists a trace such that $\psi$ holds'', which is denoted by $\exists \pi.~\psi$. The derived operators $\LTLdiamond$, $\LTLsquare$, and $\mathcal W$ are defined as for LTL.
We abbreviate the formula $\bigwedge_{x\in X} (x_\pi \leftrightarrow x_{\pi'})$, expressing that the traces $\pi$ and $\pi'$ are equal with respect to a set $X \subseteq \mathit{AP}$ of atomic propositions, by $\pi =_X \pi'$.
Furthermore, we call a trace variable $\pi$ free in a HyperLTL formula if there is no quantification over $\pi$ and we call a HyperLTL formula $\varphi$ closed if there exists no free trace variable in $\varphi$.

\paragraph{HyperLTL Semantics.}
A HyperLTL formula defines a \emph{hyperproperty}, i.e., a set of sets of traces. A set $T$ of traces satisfies the hyperproperty if it is an element of this set of sets. 
Formally, the semantics of HyperLTL formulas is given with respect to a \emph{trace assignment} $\Pi$ from $\mathcal{V}$ to $\mathit{TR}$, i.e., a partial function mapping trace variables to actual traces. $\Pi[\pi \mapsto t]$ denotes that $\pi$ is mapped to $t$, with everything else mapped according to $\Pi$. $\Pi[i,\infty]$ denotes the trace assignment that is equal to $\Pi(\pi)[i,\infty]$ for all $\pi$.
\begin{align*}
	&\Pi \models_T~\exists \pi. \psi &&\text{iff}\hspace{5ex} \text{there exists}~t \in T~:~ \Pi[\pi \mapsto t] \models_T \psi \\
	&\Pi \models_T~\forall \pi. \psi &&\text{iff}\hspace{5ex} \text{for all}~t \in T~:~ \Pi[\pi \mapsto t] \models_T \psi \\
	&\Pi \models_T~a_{\pi} &&\text{iff}\hspace{5ex} a \in \Pi(\pi)[0] \\
	&\Pi \models_T~\neg \psi &&\text{iff}\hspace{5ex} \Pi \not \models_T \psi \\
	&\Pi \models_T~\psi_1 \vee \psi_2 &&\text{iff}\hspace{5ex} \Pi \models_T \psi_1~\text{or}~\Pi \models_T \psi_2 \\
	&\Pi \models_T~\LTLnext \psi &&\text{iff}\hspace{5ex} \Pi[1,\infty] \models_T \psi \\
	&\Pi \models_T~\psi_1 \LTLuntil \psi_2 &&\text{iff}\hspace{5ex} \text{there exists}~i \geq 0 : \Pi[i,\infty] \models_T \psi_2 \\
	& &&\hspace{7.2ex} \text{and for all}~0 \leq j < i~\text{we have}~\Pi[j,\infty] \models_T \psi_1
\end{align*}
We say a set of traces $T$ \emph{satisfies} a HyperLTL formula $\varphi$ if $\Pi \models_T \varphi$, where $\Pi$ is the empty trace assignment.

\subsection{System model}
A \emph{Kripke structure} is a tuple $K=(S,s_0,\delta,\mathit{AP},L)$ consisting of a set of states $S$, an initial state $s_0 \in S$, a transition function $\delta: S \rightarrow 2^S$, a set of \emph{atomic propositions} $\mathit{AP}$, and a \emph{labeling function} $L : S \rightarrow 2^\mathit{AP}$, which labels every state with a set of atomic propositions. We assume that each state has a successor, i.e., $\delta(s) \not = \emptyset$. This ensures that every run on a Kripke structure can always be extended to an infinite run. We define a \emph{path} of a Kripke structure as an infinite sequence of states $s_0s_1\dots \in S^\omega$ such that $s_0$ is the initial state of $K$ and $s_{i+1} \in \delta(s_i)$ for every $i \in \mathbb{N}$. We denote the set of all paths of $K$ that start in a state $s$ with $\mathit{Paths}(K,s)$. Furthermore, $\mathit{Paths}^*(K,s)$ denotes the set of all path prefixes and $\mathit{Paths}^\omega(K,s)$ the set of all path suffixes. A \emph{trace} of a Kripke structure is an infinite sequence of sets of atomic propositions $L(s_0),L(s_1),\dots \in (2^\mathit{AP})^\omega$, such that $s_0$ is the initial state of $K$ and $s_{i+1} \in \delta(s_i)$ for every $i \in \mathbb{N}$.
We denote the set of all traces of $K$ that start in a state $s$ with $\mathit{TR}(K,s)$. 
%Furthermore, $\mathit{T}^*(K,s)$ denotes the set of all trace prefixes and $\mathit{Traces}^\omega(K,s)$ the set of all trace suffixes.
We say that a Kripke structure $K$ \emph{satisfies} a HyperLTL formula $\varphi$ if its set of traces satisfies $\varphi$, i.e., if  $\Pi \models_{\mathit{TR}(K,s_0)} \varphi$, where $\Pi$ is the empty trace assignment.

\subsection{Automata over infinite words}
In our construction we use automata over infinite words. A \emph{B\"uchi automaton} is a tuple $\mathcal B = (Q,Q_0,\delta,\Sigma,F)$, where $Q$ is a set of states, $Q_0$ is a set of initial states, $\delta: Q \times \Sigma \rightarrow 2^Q$ is a transition relation, and $F\subset Q$ are the accepting states. 
A run of $\mathcal B$ on an infinite word $w = \alpha_1 \alpha_2 \dots \in \Sigma^\omega$ is an infinite sequence $r = q_0 q_1 \dots \in Q^\omega$ of states, where $q_0 \in Q_0$ and for each $i \ge 0$, $q_{i+1} = \delta(q_i,\alpha_{i+1})$. 
We define $\textbf{Inf}(r)=\{q \in Q \mid \forall i \exists j>i.~r_j = q\}$.
%A run $r$ is called accepting if $\max \{c(q) \mid q \in \textbf{Inf}(r)\}$ is even.
A run $r$ is called accepting if $\textbf{Inf}(r) \cap F \not = \emptyset$.
A word $w$ is accepted by $\mathcal B$ and called a \emph{model} of $\mathcal B$ if there is an accepting run of $\mathcal B$ on $w$.

Furthermore, an \emph{alternating automaton}, whose runs generalize from sequences to trees, is a tuple $\mathcal A = (Q,Q_0,\delta,\Sigma,F)$. $Q,Q_0, \Sigma$, and $F$ are defined as above and $\delta: Q \times \Sigma \rightarrow \mathbb{B}^+{Q}$ being a transition function, which maps a state and a symbol into a Boolean combination of states. Thus, a run(-tree) of an alternating B\"uchi automaton $\mathcal A$ on an infinite word $w$ is a $Q$-labeled tree.
A word $w$ is accepted by $\mathcal{A}$ and called a \emph{model} if there exists a run-tree $T$ such that all paths $p$ trough $T$ are accepting, i.e., $\textbf{Inf}(p) \cap F \not = \emptyset$.

%Note that every (infinite) model $w$ of a (alternating) B\"uchi automaton is lasso-shaped~\cite{?}, i.e., every model $w$ can be represented as two finite words $u$ and $v$, such that $w=u(v)^\omega$. Thus, whenever we write \emph{the lasso $w$}, we refer to the finite representation $u,v$ of $w$.

%The automaton is called deterministic if the set $Q_0$ is a singleton and for each $(q,\alpha) \in Q \times \Sigma$ we have  $|\delta(q,\alpha)|\leq 1$. The automaton is called unambiguous if for each accepted word $w$ there is exactly one accepting run of the automaton on~$w$. 
%A parity automaton is called a \emph{B\"uchi} automaton if the image of $c$ is contained in $\{1,2\}$.
%An automaton is \emph{complete} if each state has an outgoing transition for each letter $\alpha \in \Sigma$. In this paper, we assume every automaton to be complete and unambiguous. 

A strongly connected component (SCC) in $\mathcal A$  is a maximal strongly connected component of the graph induced by the automaton.
%A strongly connected component is called \emph{terminal} if none of the states in the SCC has a transition, that leaves the SCC.
An SCC is called \emph{accepting} if one of its states is an accepting state in $\mathcal A$. 
%An SCC is called \emph{trivial} if it is composed of a single loop, otherwise it is a \emph{non-trivial} SCC. 

%\begin{lemma}
%	\emph{\cite{Schneider:2004:VRS:961891}} Given an LTL formula $\varphi$, one can construct an unambiguous B\"uchi automaton $B$ with $L(\varphi)=L(B)$. The number of states of $B$ is exponential in the size of~$\varphi$. 
%\end{lemma}

%\begin{lemma}\label{M&H}
%	\emph{\cite{MIYANO1984321}} For every alternating B\"uchi automaton $\mathcal A$,
%	there exists a nondeterministic B\"uchi automaton $\mathcal B$ with $\mathcal{L}(A) = \mathcal{L}(B)$ of size exponential in the size of $A$.
%\end{lemma}

%%%%%%%%%%%%%
\section{Quantitative Hyperproperties}

Quantitative Hyperproperties are properties of sets of computation traces that express a bound on the number of traces that may appear in a certain relation.
In the following, we study quantitative hyperproperties that are specified in terms of HyperLTL formulas.
We consider expressions of the following general form:
\[
\forall \pi_1,\dots,\pi_k.\ \varphi \rightarrow (\# \sigma: A.\ \psi  \triangleleft n)
\]
Both the universally quantified variables $\pi_1,\dots,\pi_k$ and the variable $\sigma$ after the \emph{counting} operator $\#$ are trace variables; $\varphi$ is a HyperLTL formula over atomic propositions $AP$ and free trace variables $\pi_1 \ldots \pi_k$;
$A \subseteq AP$ is a set of atomic propositions;
$\psi$ is a HyperLTL formula over atomic propositions $AP$ and free trace variables $\pi_1 \ldots \pi_k$ and, additionally $\sigma$.
%, where $\sigma$ may only be used for atomic propositions in $X$. 
The operator $\triangleleft \in \{<,\leq,=,>,\geq\}$ is a comparison operator; and $n \in \mathbb{N}$ is a natural number.

For a given set of traces $T$ and a valuation of the trace variables $\pi_1,\dots,\pi_k$, the term $\# \sigma: A.\, \psi$ computes the number of traces $\sigma$ in $T$ that differ in their valuation of the atomic propositions in $A$ and satisfy $\psi$.
The expression $\# \sigma: A.\, \psi  \triangleleft n$ is $\mathit{true}$ iff the resulting number satisfies the comparison with $n$.
Finally, the complete expression $\forall \pi_1,\dots,\pi_k.\,\varphi \rightarrow (\# \sigma: A.\, \psi  \triangleleft n)$ is $\mathit{true}$ iff for all combinations $\pi_1,\dots,\pi_k$ of traces in $T$ that satisfy $\varphi$, the comparison $\# \sigma: A.\, \psi  \triangleleft n$ is satisfied.

\begin{example}[Quantitative non-interference]
  \label{quantnon}
Quantitative information-flow policies~\cite{Gray/1991/TowardAMathematicalFoundationForIFSecurity,DBLP:journals/entcs/ClarkHM05,DBLP:conf/ccs/KopfB07,DBLP:journals/jcs/ClarksonMS09} allow the flow of a bounded amount of information.
One way to measure leakage is with \emph{min-entropy}~\cite{Smith/2009/OnTheFoundationsOfQantitativeInformationFlow}, which quantifies the amount of information an attacker can gain given the answer to a single guess about the secret.
The \emph{bounding problem}~\cite{Yasuoka+Terauchi/2010/OnBoundingProblemsOfQuantitativeInformationFlow} for min-entropy is to determine whether that amount is bounded from above by a constant $2^c$, corresponding to $c$ bits. 
We assume that the program whose leakage is being quantified is deterministic, and assume that the secret input to that program is uniformly distributed.
The bounding problem then reduces to determining that there is no tuple of $2^c+1$ distinguishable traces~\cite{Smith/2009/OnTheFoundationsOfQantitativeInformationFlow,Yasuoka+Terauchi/2010/OnBoundingProblemsOfQuantitativeInformationFlow}. Let $O \subseteq AP$ be the set of observable outputs. A simple quantitative information flow policy is then the following quantitative hyperproperty, which bounds the number of distinguishable outputs to $2^c$, corresponding to a bound of $c$ bits of information:
	$$
		\CQuan \sigma: O.\ \mathit{true} \leq 2^c
	$$
                A slightly more complicated information flow policy is quantitative non-interference. In quantitative non-interference, the bound must be satisfied for every individual input. Let $I \subseteq AP$ be the observable inputs to the system. Quantitative non-interference is the following quantitative hyperproperty\footnote{We write $\pi =_A \pi'$ short for $\pi_A = \pi'_A$ where $\pi_A$ is the $A$-projection of $\pi$}:
        $$
	\forall \pi.\, \CQuan \sigma\colon O.\ ( \LTLsquare({\pi} =_I \sigma) ) \leq 2^c
	$$
For each trace $\pi$ in the system, the property checks whether there are more than $2^c$ traces $\sigma$ that have the same observable input as $\pi$ but different observable output.
\end{example}

\begin{example}[Deniability]
A program satisfies \emph{deniability} (see, for example,~\cite{DBLP:journals/pvldb/BindschaedlerSG17,DBLP:journals/popets/ChakrabortiCS17}) when there is no proof that a certain input occurred from simply observing the output, i.e., given an output of a program one cannot derive the input that lead to this output. A deterministic program satisfies deniability when each output can be mapped to at least two inputs. A quantitative variant of deniability is when we require that the number of corresponding inputs is larger than a given threshold. Quantitative deniability can be specified as the following quantitative Hyperproperty:
 $$
\forall \pi.\, \CQuan \sigma \colon I.\, (\LTLsquare (\pi =_O \sigma))  > n
$$
For all traces $\pi$ of the system we count the number of sequences $\sigma$ in the system with different input sequences and the same output sequence of $\pi$, i.e., for the fixed output sequence given by $\pi $ we count the number of input sequences that lead to this output. 
\end{example}

\section{Model Checking Quantitative Hyperproperties}
We present a model checking algorithm for quantitative hyperproperties based on model counting. 
The advantage of the algorithm is that its runtime complexity is independent of the bound $n$ and thus avoids the $n$-fold self-composition necessary for any encoding of the quantitative hyperproperty in \hyperltl.  

%In our new algorithm we use a model counting algorithm to count the number of traces satisfying the quantitative-hyperproperty which requires a counter of $\log(n)$ bits. 
Before introducing our novel counting-based algorithm, we start by a translation of quantitative hyperproperties into formulas in \hyperltl and establishing an exponential lower bound for its representation.
%We continue then with introduction of the counting-based algorithm. \todo{exponential lower bound?}

\subsection{Standard model checking algorithm: encoding quantitative hyperproperties in HyperLTL}
The idea of the reduction is to check a lower bound of $n$
traces by existentially quantifying over $n$ traces, and to check an
upper bound of $n$ traces by \emph{universally} quantifying over $n+1$ traces.
The resulting HyperLTL formula can be verified using the standard
model checking algorithm for HyperLTL~\cite{conf/post/ClarksonFKMRS14}.

\begin{theorem}
  Every quantitative hyperproperty
  $
  \forall \pi_1,\dots,\pi_k.\ \psi_\iota \rightarrow (\# \sigma: A.\ \psi  \triangleleft n)
  $
  can be expressed as a HyperLTL formula. For $\triangleleft \in \{\leq\} (\{<\})$, the HyperLTL formula has $n+k+1 (\text{resp. } n+k)$ universal trace quantifiers in addition to the quantifiers in $\psi_\iota$ and $\psi$.
  For $\triangleleft \in \{\geq\} (\{>\})$, the HyperLTL formula has $k$ universal trace quantifiers and $n$ $(\text{resp. } n+1)$ existential trace quantifiers in addition to the quantifiers in $\psi_\iota$ and $\psi$. For $\triangleleft \in \{=\}$, the HyperLTL formula has $n+k+1$ universal trace quantifiers and $n$ existential trace quantifiers in addition to the quantifiers in $\psi_\iota$ and $\psi$.
\end{theorem}
\begin{proof}
For $\triangleleft \in \{\leq \}$, we encode the quantitative hyperproperty $\forall \pi_1,\dots,\pi_k.\ \psi_\iota \rightarrow (\# \sigma: A.\ \psi \triangleleft n)$ as the following HyperLTL formula:
\[
  \forall \pi_1,\dots,\pi_k.\ \forall \pi'_1, \ldots, \pi'_{n+1}.\ \left( \psi_\iota \wedge \bigwedge_{i\neq j} \LTLdiamond({\pi'_i}\neq_A {\pi'_j})\right) \rightarrow \left(\bigvee_i \neg \psi[\sigma \mapsto \pi'_i]\right)  
\]
where $\psi[\sigma \mapsto \pi'_i]$ is the HyperLTL formula $\psi$ with all occurrences of $\sigma$ replaced by $\pi'_i$.
The formula states that there is no tuple of $n+1$ traces $\pi'_1, \ldots, \pi'_{n+1}$ different in the evaluation of $A$, that satisfy $\psi$. In other words, for every $n+1$ tuple of traces $\pi'_1, \ldots, \pi'_{n+1}$ that differ in the evaluation of $A$, one of the paths must violate $\psi$.
For $\triangleleft \in \{ <\}$, we use the same formula, with $\forall \pi'_1, \ldots, \pi'_{n}$ instead of $\forall \pi'_1, \ldots, \pi'_{n+1}$.

For  $\triangleleft \in \{\geq \}$, we encode the quantitative hyperproperty analogously as the HyperLTL formula
\[
  \forall \pi_1,\dots,\pi_{k}.\ \exists \pi'_1, \ldots, \pi'_{n}.\ \psi_\iota \rightarrow \left(  \bigwedge_{i\neq j} \LTLdiamond({\pi'_i}\neq_A {\pi'_j})\right) \wedge \left(\bigwedge_i \psi[\sigma \mapsto \pi'_i]\right)  
\]
The formula states that there exist paths $\pi'_1, \ldots, \pi'_{n}$ that differ in the evaluation of $A$ and that all satisfy $\psi$.
For $\triangleleft \in \{ >\}$, we use the same formula, with $\exists \pi'_1, \ldots, \pi'_{n+1}$ instead of $\forall \pi'_1, \ldots, \pi'_{n}$.
Lastly, for $\triangleleft \in \{ = \}$, we encode the quantitative hyperproperty as a conjunction of the encodings for $\leq$ and for $\geq$.
%
%	.
\end{proof}

\begin{example}[Quantitative non-interference in HyperLTL]
	 As discussed in Example~\ref{quantnon}, quantitative non-interference is the quantitative hyperproperty 
	 $$\forall \pi.\, \CQuan \sigma\colon O. \, \LTLsquare({\pi} =_I \sigma) \leq 2^c,$$ where we measure the amount of leakage with min-entropy~\cite{Smith/2009/OnTheFoundationsOfQantitativeInformationFlow}.
	The bounding problem for min-entropy asks whether the amount of information leaked by a system is bounded by a constant $2^c$ where $c$ is the number of bits. This is encoded in HyperLTL as the requirement that there are no $2^{c}+1$ traces distinguishable in their output:
		\[
	%\mathit{QN}(c)\coloneqq
	\forall \pi_0.\ \forall \pi_1 \ldots \forall \pi_{2^c}.~  \left(\bigwedge_i \LTLsquare ({\pi_i} =_I {\pi_0}) \right) \rightarrow \left(\bigvee_{i \neq j} \LTLsquare ({\pi_i} =_O {\pi_j})\right).
	\]
        This formula is equivalent to the formalization of quantitative non-interference given in \cite{conf/cav/FinkbeinerRS15}.
\end{example}

Model checking quantitative hyperproperties via the reduction to HyperLTL is very expensive. In the best case, when $\triangleleft \in \{\leq, < \}$, $\psi_\iota$ does not contain existential quantifiers, and $\psi$ does not contain universal quantifiers, we obtain an HyperLTL formula without quantifier alternations, where the number of quantifiers grows linearly with the bound $n$. For $m$ quantifiers, the HyperLTL model checking algorithm~\cite{conf/cav/FinkbeinerRS15} constructs and analyzes the $m$-fold self-composition of the Kripke structure. The running time of the model checking algorithm is thus exponential in the bound. If  $\triangleleft \in \{\geq, >, = \}$, the encoding additionally introduces a quantifier alternation. The model checking algorithm checks quantifier alternations via a complementation of B\"uchi automata, which adds another exponent, resulting in an overall doubly exponential running time.

The model checking algorithm we introduce in the next section avoids the $n$-fold self-composition needed in the model checking algorithm of HyperLTL and its complexity is independent of the bound $n$.

%%%%%%%%%%%%%

%%%%%%%%%%%%%
\subsection{Counting-based model checking algorithm}
A Kripke structure $K= (S,s_0, \tau, \AP, L)$ violates a quantitative hyperproperty
$$\varphi = \forall \pi_1, \dots, \pi_k.~ \psi_\iota \rightarrow (\CQuan \sigma : A. \psi \triangleleft n )$$ 
if there is a $k$-tuple $t=(\pi_1,\dots,\pi_k)$ of traces $\pi_i \in \TR(K)$ that satisfies the formula
\[
  \exists \pi_1,\dots,\pi_k.\ \psi_\iota \wedge (\# \sigma: A.\ \psi\,  \overline{\triangleleft}\, n)
  \]
where $\overline{\triangleleft}$ is the negation of the comparison operator $\triangleleft$. The tuple~$t$ then satisfies the property $\psi_\iota$ and the number of $(k+1)$-tuples $t'=(\pi_1,\dots, \pi_k, \sigma)$ for $\sigma \in \TR(K)$  that satisfy $\psi$ and differ pairwise in the $A$-projection of $\sigma$ satisfies the comparison $\overline\triangleleft~n$ (The $A$-projection of a sequence $\sigma$ is defined as the sequence $\sigma_A \in (2^A)^\omega$, such that for every position $i$ and every $a \in A$ it holds that $a \in \sigma_A[i]$ if and only if  $a \in \sigma[i]$).
The tuples $t'$ can be captured by the automaton composed of the product of an automaton $A_{\psi_\iota \wedge \psi}$ that accepts all $k+1$ of traces that satisfy both $\psi_\iota$ and $\psi$ and a $k+1$-self composition of $K$. Each accepting run of the product automaton presents $k+1$ traces of $K$ that satisfy $\psi_\iota \wedge \psi$. On top of the product automaton, we apply a special counting algorithm which we explain in detail in Section~\ref{sec:countingalgo} and check if the result satisfies the comparison~$\overline\triangleleft~n$.

Algorithm~\ref{alg:mcqhyper} gives a general picture of our model checking algorithm. The algorithm has two parts. The first part applies if the relation $\overline \triangleleft$ is one of $\{\ge , >\}$. In this case, the algorithm checks whether a sequence over $\AP_\psi$ (propositions in $\psi$) corresponds to infinitely many sequences over $A$. This is done by checking whether the product automaton $B$ has a so-called  \emph{doubly pumped lasso}(DPL), a subgraph with two connected lassos, with a unique sequence over $\AP_\psi$ and different sequences over $A$. Such a doubly pumped lasso matches the same sequence over $\AP_\psi$ with infinitely many sequences over $A$ (more in section~\ref{sec:countingalgo}). If no doubly pumped lasso is found, a projected model counting algorithm is applied in the second part of the algorithm in order to compute either the maximum or the minimum value, corresponding to the comparison operator $\overline \triangleleft$. In the next subsections, we explain the individual parts of the algorithm in detail.
\begin{algorithm}[t]
\scriptsize
\begin{algorithmic}[1]
	\REQUIRE Quantitative Hyperproperty $\varphi =\forall \pi_1 \dots \pi_k.~ \psi_\iota \rightarrow (\CQuan \sigma : A. \psi \triangleleft n )$, Kripke Structure $K= (S,s_0, \tau, \AP, L)$
	\ENSURE $K \models \varphi$
	\STATE $B =\textit{QHLTL2BA}(K,\pi_1,\dots,\pi_k,\psi_\iota \wedge \psi)$
	\STATE /*Check Infinity*/	
	\IF{$\overline \triangleleft \in \{\ge, >\} $}
	\STATE $\textit{ce} = \textit{DPL}(B)$
	\IF{$\textit{ce} \not = \bot$}
		\RETURN \textit{ce}
	\ENDIF
	\ENDIF
	\STATE /*Apply Projected Counting Algorithm*/
	\IF{$\overline \triangleleft \in \{\ge, >\}$}
		\STATE $\textit{ce} = \textit{MaxCount}(B,n,\overline\triangleleft)$
	\ELSE
		\STATE $\textit{ce} = \textit{MinCount}(B,n,\overline\triangleleft)$
	\ENDIF
	\RETURN \textit{ce}
\end{algorithmic}
\caption{\small Counting-based Model Checking of Quantitative Hyperproperties}		
\label{alg:mcqhyper}
\end{algorithm}

\subsection{B\"uchi automata for quantitative hyperproperties} 
For a quantitative hyperproperty $\varphi= \forall \pi_1 \dots \pi_k.~ \psi_\iota \rightarrow (\CQuan \sigma : A. \psi \triangleleft n )$ and a Kripke structure $K=(S,s_0, \tau, \AP, L)$, we first construct an alternating automaton $A_{\psi_{\iota} \wedge \psi}$ for the \hyperltl property $\psi_{\iota} \wedge \psi$. 
Let $A_{\psi_1} = (Q_1,q_{0,1}, \Sigma_2, \delta_1, F_1)$ and $A_{\psi_2} = (Q_2,q_{0,2}, \Sigma_2, \delta_2, F_2)$ be alternating automata for subformulas $\psi_1$ and $\psi_2$. Let $\Sigma= 2^{\AP_\varphi}$ where $AP_\varphi$ are all indexed atomic propositions that appear in $\varphi$. $A_{\psi_{\iota} \wedge \psi}$ is constructed using following rules\footnote{The construction
follows the one presented in  \cite{conf/cav/FinkbeinerRS15} with a slight modification on the labeling of transitions. Labeling over atomic proposition instead of the states of the Kripke structure suffices, as any nondeterminism in the Kripke structure is inherently resolved, because we quantify over trace not paths}: 
\begin{center}
\scriptsize
\begin{tabular}{|c||c|}
\hline
$\varphi = a_\pi$ & $A_\varphi = (\{q_0\},q_0,\Sigma,\delta,\emptyset)$ where $\delta(q_0,\alpha)= (a_\pi \in \alpha)$\\	
\hline
$\varphi = \neg a_\pi$ & $A_\varphi = (\{q_0\},q_0,\Sigma,\delta,\emptyset)$ where $\delta(q_0,\alpha)= (a_\pi \not \in \alpha)$\\	
\hline
$\varphi = \psi_1 \wedge \psi_2$ & $A_\varphi = (Q_1 \cupdot Q_2 \cupdot \{q_0\}, q_0, \Sigma, \delta, F_1 \cupdot F_2)$ \\ 
& where $\delta(q,\alpha) = \delta_1(q_{0,1},\alpha) \wedge \delta_2(q_{0,2},\alpha)$\\
& \quad and $\delta(q,\alpha) = \delta_i(q,\alpha)$ when $q\in Q_i$ for $i\in \{1,2\}$\quad \quad \\
\hline
$\varphi = \psi_1 \vee \psi_2$ & $A_\varphi = (Q_1 \cupdot Q_2 \cupdot \{q_0\}, q_0, \Sigma, \delta, F_1 \cupdot F_2)$ \\ 
& where $\delta(q,\alpha) = \delta_1(q_{0,1},\alpha) \vee \delta_2(q_{0,2},\alpha)$\\
& and $\delta(q,\alpha) = \delta_i(q,\alpha)$ when $q\in Q_i$ for $i\in \{1,2\}$\\
\hline
$\varphi = \LTLcircle \psi_1$ & $A_\varphi = (Q_1 \cupdot\{q_0\}, q_0, \Sigma, \delta, F_1)$  \\
& where $\delta(q,\alpha) = q_{0,1}$ \\
& and $\delta(q,\alpha) = \delta_1(q,\alpha) $ for $q \in Q_1$\\
\hline
$\varphi = \psi_1 \LTLuntil \psi_2$ & $A_\varphi = (Q_1 \cupdot Q_2 \cupdot \{q_0\}, q_0, \Sigma, \delta, F_1 \cupdot F_2)$\\
& where $\delta(q_0,\alpha) = \delta_2(q_{0,2}, \alpha) \vee (\delta_1(q_{0,1},\alpha) \wedge q_0)$\\
& and $\delta(q,\alpha) = \delta_i(q,\alpha)$ when $q\in Q_i$ for $i\in \{1,2\}$\\
\hline
$\varphi = \psi_1 \LTLrelease \psi_2$ & $A_\varphi = (Q_1 \cupdot Q_2 \cupdot \{q_0\}, q_0, \Sigma, \delta, F_1 \cupdot F_2 \cupdot \{q_0\})$\\
& where $\delta(q_0,\alpha) = \delta_2(q_{0,2}, \alpha) \wedge (\delta_1(q_{0,1},\alpha) \vee q_0)$\\
& and $\delta(q,\alpha) = \delta_i(q,\alpha)$ when $q\in Q_i$ for $i\in \{1,2\}$\\
\hline
\end{tabular}
\end{center}
For a quantified formula $\varphi = \exists\pi. \psi_1$, we construct the product automaton of the Kripke structure $K$ and the B\"uchi automaton  of $\psi_1$. Here we reduce the alphabet of the automaton by projecting all atomic proposition in $\AP_\pi$ away: 
\begin{center}
\scriptsize
\begin{tabular}{|c||c|}
\hline
	$\varphi = \exists \pi. \psi_1$ & $A_\varphi = (Q_1 \times S \cupdot \{q_0\}, \Sigma \setminus \AP_\pi,\delta,  F_1 \times S ) $\\
& where $\delta(q_0, \alpha ) = \{(q',s') \mid q' \in \delta_1(q_{0,1},\alpha \cup \alpha'), s' \in \tau(s_0), (L(s_0))_\pi =_{\AP_\pi} \alpha'\}$\\
& and $\delta((q,s),\alpha) = \{(q',s') \mid  q' \in \delta_1(q, \alpha\cup \alpha'), s' \in \tau(s), (L(s))_\pi =_{\AP_\pi} \alpha' \}$\\
\hline

\end{tabular}
\end{center}
Given the B\"uchi automaton for the hyperproperty $\psi_\iota \wedge \psi$ it remains to construct the product with the $k+1$-self composition of $K$. The transitions of the automaton are defined over labels from $\Sigma= 2^{\AP^*}$ where $AP^* ={\AP_\sigma \cup \bigcup_i \AP_{\pi_i}}$. $A_{\psi_{\iota} \wedge \psi}$. This is necessary to identify which transition was taken in each copy of $K$, thus, mirroring a tuple of traces in $K$. For each of the variables $\pi_1,\dots \pi_k$ and $\sigma$ we use following rule: 
\begin{center}
\scriptsize
\begin{tabular}{|c||c|}
\hline
	$\varphi = \exists \pi. \psi_1$ & $A_\varphi = (Q_1 \times S \cupdot \{q_0\}, \Sigma,\delta,  F_1 \times S ) $\\
& where $\delta(q_0, \alpha ) = \{(q',s') \mid q' \in \delta_1(q_{0,1},\alpha), s' \in \tau(s_0), (L(s_0))_\pi =_{\AP_\pi} \}$\\
& and $\delta((q,s),\alpha) = \{(q',s') \mid  q' \in \delta_1(q, \alpha), s' \in \tau(s), (L(s))_\pi =_{\AP_\pi}  \}$\\
\hline

\end{tabular}
\end{center}
Finally, we transform the resulting alternating automaton to an equivalent B\"uchi automaton following the construction of Miyano and Hayashi~\cite{MIYANO1984321}.

\subsection{Counting models of $\omega$-Automata}
\label{sec:countingalgo}
Computing the number of words accepted by a B\"uchi automaton can be done by examining its accepting lassos. 
Consider, for example, the B\"uchi automata over the alphabet $2^{\{a\}}$ in Figure~\ref{fig:automodels}. The automaton on the left  has one accepting lasso $(q_0)^\omega$ and thus has only one model, namely $\{a\}^\omega$. The automaton on the right  has infinitely many  accepting lassos $(q_0\{\})^i\{a\}(q_1(\{\} \vee \{a\}))^\omega$ that accept infinitely many different words all of the from $\{\}^*\{a\}(\{\}\vee \{a\})^\omega$. 
\begin{figure}
\centering
\begin{tikzpicture}
\node(a1)[circle,double,draw,initial]at(0,0){$q_0$};
\node(a2)[circle,draw]at(2,0){$q_1$}; 
\path[->,draw](a1) edge [loop above] node{$a$}(a1);
\path[->,draw](a1) edge node[above]{$\neg a$} (a2);
\path[->,draw](a2) edge [loop above] node{$*$}(a2);
\end{tikzpicture}
\hspace{2cm}
\begin{tikzpicture}
\node(a1)[circle,draw,initial]at(0,0){$q_0$};
\node(a2)[circle,double,draw]at(2,0){$q_1$}; 
\path[->,draw](a1) edge [loop above] node{$\neg a$}(a1);
\path[->,draw](a1) edge node[above]{$a$} (a2);
\path[->,draw](a2) edge [loop above] node{$*$}(a2);
\end{tikzpicture}
%%%%%%%%%%%%%%%%%%%%%%%%%%%%%
\caption[]{\scriptsize B\"uchi automata with one model (left) and infinitely many models (right).}
\label{fig:automodels}
\end{figure}
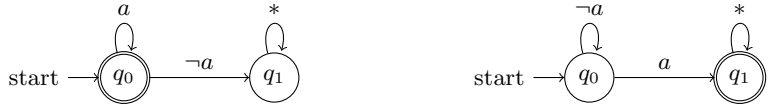
Computing the models of a B\"uchi automaton is insufficient for model checking quantitative hyperproperties as we are not interested in the total number of models. We rather \emph{maximize}, respectively \emph{minimize}, over sequences of subsets of atomic propositions \emph{the number of projected models} of the B\"uchi automaton. For instance, consider the automaton given in Figure~\ref{fig:autohypermodels}.
\begin{figure}
\centering
\begin{tikzpicture}
\node(a1)[circle,draw,initial]at(0,0){$q_0$};
\node(a2)[circle,double,draw]at(2,0){$q_1$}; 
\path[->,draw](a1) edge [loop above] node{$\neg a \wedge b$}(a1);
\path[->,draw](a1) edge node[above]{$a$} (a2);
\path[->,draw](a2) edge [loop above] node{$b$}(a2);
\end{tikzpicture}
\caption[]{\scriptsize A two-state B\"uchi automaton, such that there exist exactly two $\{b\}$-projected models for each $\{a\}$-projected sequence.}
\label{fig:autohypermodels}
\end{figure}
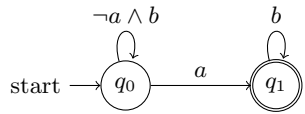
The automaton has infinitely many models. However, the maximum number of sequences $\sigma_b \in 2^{\{b\}}$ that correspond to accepting lassos in the automaton with a unique sequence $\sigma_a \in 2^{\{a\}}$ is two: For example, let $n$ be a natural number. For any model of the automaton and for each sequence $\sigma_a := \{\}^n\{a\}(\{\})^\omega$ the automaton accepts the following two sequences: $\{b\}^n\{\}\{b\}^\omega$ and $\{b\}^\omega$. 
Formally, given a B\"uchi automaton $\mathcal{B}$ over $\mathit{AP}$ and a set $A$, such that $A \subseteq \mathit{AP}$, an \emph{$A$-projected model} (or projected model over $A$) is defined as a sequence $\sigma_A \in (2^\mathit{A})^\omega$ that results in the $A$-projection of an accepting sequence $\sigma \in (2^\mathit{AP})^\omega$.

In the following, we define the maximum model counting problem over automata and give an algorithm for solving the problem. We show  how to use the algorithm for model checking quantitative hyperproperties. 

\begin{definition}[Maximum Model Counting over Automata (MMCA)]
	Given a B\"uchi automaton $B$ over an alphabet $2^\AP$ for some set of atomic propositions $\AP$ and sets $X,Y,Z \subseteq \AP$ the maximum model counting problem is to compute $$\max \limits_{\sigma_Y \in (2^{Y})^\omega} |\{\sigma_X \in (2^{X})^\omega \mid \exists \sigma_Z \in(2^{Z})^\omega.~ \sigma_X\cup\sigma_Y\cup\sigma_Z \in L(B)\}|$$
	where $\sigma \cup \sigma'$ is the point-wise union of $\sigma $ and $\sigma'$. 
\end{definition}
%\todo{define projected models}
As a first step in our algorithm, we show how to check whether the maximum model count is equal to infinity. 
\begin{definition}[Doubly Pumped Lasso]
For a graph $G$, a doubly pumped lasso in $G$ is a subgraph that entails a cycles $C_1$ and another different cycle $C_2$ that is reachable from $C_1$. 
\end{definition} 
\vskip -1cm 
\begin{figure}
\centering
\begin{tikzpicture}
\coordinate(p1)at(0,0);
\fill (p1) circle (1pt);
\coordinate(p2)at(0.25,0);
\fill (p2) circle (1pt);
\draw[->](p1)--(p2);
\coordinate(p3)at(0.5,0);
\fill (p3) circle (1pt);
\draw[->](p2)--(p3);
\coordinate(p4)at(0.75,0);
\fill (p4) circle (1pt);
\draw[->](p3)--(p4);
\coordinate(p5)at(1,0.25);
\fill (p5) circle (1pt);
\draw[->](p4)--(p5);
\coordinate(p6)at(0.75,.5);
\fill (p6) circle (1pt);
\draw[->](p5)--(p6);
\coordinate(p7)at(0.5,.5);
\fill (p7) circle (1pt);
\draw[->](p6)--(p7);
\coordinate(p8)at(0.25,.25);
\fill (p8) circle (1pt);
\draw[->](p7)--(p8);
\draw[->](p8)--(p3);
\coordinate(connect)at(1.17,0);
\fill (connect) circle (1pt);
\draw[->](p4)--(connect);
\coordinate(q3)at(1.5,0);
\fill (q3) circle (1pt);
\draw[->](connect)--(q3);
\coordinate(q4)at(1.75,0);
\fill (q4) circle (1pt);
\draw[->](q3)--(q4);
\coordinate(q5)at(2,0.25);
\fill (q5) circle (1pt);
\draw[->](q4)--(q5);
\coordinate(q6)at(1.75,.5);
\fill (q6) circle (1pt);
\draw[->](q5)--(q6);
\coordinate(q7)at(1.5,.5);
\fill (q7) circle (1pt);
\draw[->](q6)--(q7);
\coordinate(q8)at(1.25,.25);
\fill (q8) circle (1pt);
\draw[->](q7)--(q8);
\draw[->](q8)--(q3);
\end{tikzpicture}
%%%%%%%%%%%%%%%%%%%%%%%%
\hspace{1cm}
%%%%%%%%%%%%%%%%%%%%%%%%
\begin{tikzpicture}
\coordinate(p1)at(0,0);
\fill (p1) circle (1pt);
\coordinate(p2)at(0.25,0);
\fill (p2) circle (1pt);
\draw[->](p1)--(p2);
\coordinate(p3)at(0.5,0);
\fill (p3) circle (1pt);
\draw[->](p2)--(p3);
\coordinate(p4)at(0.75,0);
\fill (p4) circle (1pt);
\draw[->](p3)--(p4);
\coordinate(p5)at(1,0.25);
\fill (p5) circle (1pt);
\draw[->](p4)--(p5);
\coordinate(p6)at(0.75,.5);
\fill (p6) circle (1pt);
\draw[->](p5)--(p6);
\coordinate(p7)at(0.5,.5);
\fill (p7) circle (1pt);
\draw[->](p6)--(p7);
\coordinate(p8)at(0.25,.25);
\fill (p8) circle (1pt);
\draw[->](p7)--(p8);
\draw[->](p8)--(p3);
\coordinate(connect)at(1,0);
\fill (connect) circle (1pt);
\draw[->](p4)--(connect);
\coordinate(q5)at(1.25,-0.25);
\fill (q5) circle (1pt);
\draw[->](connect)--(q5);
\coordinate(q6)at(1,-.5);
\fill (q6) circle (1pt);
\draw[->](q5)--(q6);
\coordinate(q7)at(0.75,-.5);
\fill (q7) circle (1pt);
\draw[->](q6)--(q7);
\coordinate(q8)at(0.5,-.25);
\fill (q8) circle (1pt);
\draw[->](q7)--(q8);
\draw[->](q8)--(p3);
\end{tikzpicture}
\caption{\scriptsize Forms of doubly pumped lassos.}
\label{fig:pumpedlassos}
\end{figure}
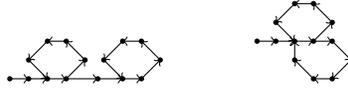
In general, we distinguish between two types of doubly pumped lassos as shown in Figure~\ref{fig:pumpedlassos}.
We call the lassos with periods $C_1$ and $C_2$ the lassos of the doubly pumped lasso.  
A doubly pumped lasso of a B\"uchi automaton $B$ is one in the graph structure of $B$. The doubly pumped lasso is called accepting when $C_2$ has an accepting state. A more generalized formalization of this idea  is given in the following theorem.  
 
\begin{theorem}
	Let $B=(Q,q_0,\delta, 2^\AP, F)$ be a B\"uchi automaton for some set of atomic propositions $\AP = X \cup Y \cup Z$ and let $\sigma' \in ( 2^{Y})^\omega$. The automaton~$B$ has infinitely many $X \cup Y$-projected models $\sigma$ with $\sigma =_{Y} \sigma'$ if and only if $B$ has an accepting doubly pumped lasso with lassos $\rho$ and $\rho'$ such that:
		%\begin{enumerate}
			1) $\rho$ is an accepting lasso
		    2) $\tr(\rho)=_{Y} \tr(\rho')=_{Y}\sigma'$
			3) The period of $\rho'$ shares at least one state with $\rho$
			and 4) $\tr(\rho)\not =_X\tr(\rho')$. %\footnote{The proof can be found in the appendix}%$\exists i'.~ \alpha_{(l+1 +i' \mod k')+j'-1} \not= \alpha_{(l+1 +i' \mod k)+j-1}$
		%\end{enumerate} 
\label{theo:infmodels}
\end{theorem}
To check whether there is a sequence $\sigma' \in (2^Y)^\omega$ such that the number of $X\cup Y$-projected models $\sigma$ of $B$ with $\sigma =_Y \sigma'$ is infinite, we search for a doubly pumped lasso satisfying the constraints given in Theorem~\ref{theo:infmodels}. This can be done by applying the following procedure:

Given a B\"uchi automaton $B=(Q,q_0,2^\AP,\delta,F)$ and sets $X,Y,Z \subseteq \AP$, we construct the following product automaton $B_\times =(Q_\times,q_{\times,0},2^\AP\times 2^\AP, \delta_\times, F_\times)$ where:
%\begin{itemize}
	%\item
	$Q_\times = Q \times Q$,
	%\item 
	$q_{\times,0} = (q_0,q_0)$,
	%\item 
	$\delta_\times = \{(s_1,s_2) \xrightarrow{(\alpha,\alpha')} (s'_1,s_2') \mid s_1 \xrightarrow{\alpha} s_2, s'_1 \xrightarrow{\alpha'} s'_2, \alpha =_{Y} \alpha'\}$ and
	%\item 
	$F_\times = Q\times F$.
%\end{itemize}
The automaton $B$ has infinitely many models $\sigma'$ if there is an accepting lasso $\rho = (q_0,q_0)(\alpha_1,\alpha'_1) \dots ((q_j,q_j')(\alpha_{j+1},\alpha'_{j+1})$ $\dots (q_k,q'_k) (\alpha_{k+1},\alpha'_{k+1}))$ in $B_\times$ such that:
%\begin{itemize}
	 $\exists h\leq j.~ q'_h = q_j$, i.e., $B$ has lassos $\rho_1$ and $\rho_2$ that share a state in the period of $\rho_1$
	and $\exists h>j.~ \alpha_h \not =_X \alpha'_h$, i.e., the lassos differ in the evaluation of $X$ in a position after the shared state and thus allows infinitely many different sequence over $X$ for the a sequence over $Y$. 
%\end{itemize} 
The lasso $\rho$ simulates a doubly pumped lasso in $B$ satisfying the constraints of Theorem~\ref{theo:infmodels}.

\begin{theorem}
	Given an alternating B\"uchi automaton $A=(Q,q_0,\delta, 2^\AP, F)$ for a set of atomic propositions $\AP = X\cup Y\cup Z$, the problem of checking whether there is a sequence  $\sigma' \in ( 2^{Y})^\omega$ such that $A$ has infinitely many $X\cup Y$-projected models  $\sigma$ with $\sigma =_Y \sigma'$ is \pspace-complete. 
\label{theo:checkinfmodels}
\end{theorem}
The lower and upper bound for the problem can be given by a reduction from and to the satisfiability problem of LTL \cite{Baier:2008:PMC:1373322}.
%\todo{explain intuitively}
%\begin{corollary}
%	 Let $B=(Q,Q_0,\delta,F)$ be a B\"uchi automaton. Then $|L(B)|= \infty$ if and only if $B$ has a non-trivial accepting SCC with at least two non-equal trivial SCCs. 
%\end{corollary}
Due to the finite structure of B\"uchi automata, if the number of models  of the automaton exceed the exponential bound $2^{|Q|}$, where $Q$ is the set of states, then the automaton has infinitely many models. 
\begin{lemma}
\label{lem:maxnummodels}
	For any B\"uchi automaton $B$, the number of models of $B$ is less or equal to $2^{|Q|}$ otherwise it is $\infty$. 
\end{lemma}
\begin{proof}
Assume the number of models is larger than $2^{|Q|}$ then there are  more than $2^{|Q|}$ accepting lassos in $B$. By the pigeonhole principle, two of them share the same $2^{|Q|}$-prefix. Thus, either they are equal or we found doubly pumped lasso in $B$.  	
\end{proof}

\begin{corollary}
	Let a B\"uchi automaton $B$ over a set of atomic propositions $\AP$ and sets $X,Y \subseteq \AP$. For each sequence $\sigma_Y \in (2^{Y})^\omega$ the number of $X\cup Y$-projected models $\sigma $ with $\sigma =_Y \sigma_Y $ is less or equal than $2^{|Q|}$ otherwise it is $\infty$. 
\label{cor:exporinfinity}
\end{corollary}
%If a B\"uchi automaton does not accept infinitely many words then we can count the words the maximum model counting algorithm presented in Algorithm~\ref{alg:counting}.
From Corollary~\ref{cor:exporinfinity}, we know that if no sequence $\sigma_Y \in (2^Y)^\omega$ matches to infinitely many $X \cup Y$-projected models then the number of such models is bound by $2^{|Q|}$. Each of these models has a run in $B$ which ends in an accepting strongly connected component. Also from Corollary~\ref{cor:exporinfinity}, we know that every model has  a lasso run of length $|Q|$. 
% Thus it suffices to count the number of finite $X\cup Y $-projected words that agree with $\sigma_Y$ that end in an accepting strongly connected component. A strongly connected component is reached in at most $|Q|$ steps from the initial state. 
For each finite sequence $w_Y$ of length $|w_Y|= |Q|$ that reaches an accepting strongly connected component, we count the number $X \cup Y$-projected words $w$ of length $|Q|$ with $w =_Y w_Y$ and that end in an accepting strongly connected component. This number is equal to the maximum model counting number.

\begin{figure}[t]
\begin{minipage}[]{.52\textwidth}
\begin{algorithm}[H]
\scriptsize 
\begin{algorithmic}[1]
	\REQUIRE $B = (Q,q_0,2^\AP,\delta,F)$, disjoint $X,Y,Z\subseteq \AP$, $n\in \nats$ 
	\ENSURE $\#_{X,Y,Z}(B)>n$
	\STATE $SCC = acceptingSCC(B) $
	\STATE $i=1$
	\STATE $W = \bigcup \limits_{S \in SCC} S$
	\WHILE{$i\leq|Q|$}
		\STATE $i = i+1$
		\FOR{$q\in W$}
			\FOR{$(q',\alpha,q)$}
				\STATE $W'= W'\cup \{q'\}$
				\FOR{$\sigma\in \Pi(q)$}
					\STATE $\Pi'(q') = \Pi'(q') \cup \{\alpha_{Y\cup X}\cdot \sigma\}$
				\ENDFOR
			\ENDFOR
		\ENDFOR
		\STATE $W = W'$
		\STATE $W' = \emptyset$
		\FOR{$q \in W$}
			\STATE $\Pi'(q)= \emptyset$
		\ENDFOR
	\ENDWHILE 
	\RETURN $\max_{X,Y,Z} \Pi(q_0)\ge n$
\end{algorithmic}
\caption{Maximum Model Counting}	
\label{alg:counting}
\end{algorithm}
\end{minipage}\begin{minipage}[]{.4\textwidth}
\begin{figure}[H]
%	\centering
	\begin{tikzpicture}
	\node[state](q3) {$q_3$};
	\node[state](q1) [above right=1cm and 2cm of q3]  {$q_1$};
	\node[state](q2) [below right=1cm and 2cm of q3] {$q_2$};
	
	\node[right of=q1] (dot1) {$\ldots$};
	\node[right of=q2] (dot2) {$\ldots$};
	\node[left of=q3] (dot3) {$\ldots$};
	
	\node[above of=q1] (set1) {$\Pi(q_1):=\{\sigma_1,\ldots,\sigma_k\}$};
	\node[below of=q2] (set2) {$\Pi(q_2):=\{\sigma'_1,\ldots,\sigma'_j\}$};
	\node[right=0.4cm] (set3) {$\Pi(q_3):=\{\alpha_1\sigma_1,\ldots,\alpha_1\sigma_k\}$};
	\node[below right=0.2cm] (set4) {$~~~~~~~~~~~\cup \{\alpha_2\sigma'_1,\ldots,\alpha_2\sigma_j\}$};
	
	\path [->] (q3) edge [bend left] node [below] {$\alpha_1$} (q1);
	\path [->] (q3) edge [bend right] node [above] {$\alpha_2$} (q2);
	\end{tikzpicture}
\end{figure}
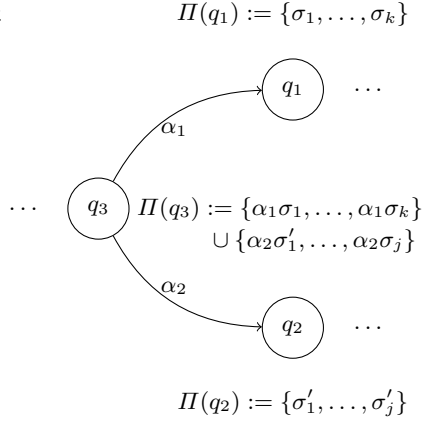
\end{minipage}
\caption{\scriptsize Maximum Model Counting Algorithm (left) and a Sketch of a step in this algorithm (right): Current elements of our working set are $q_1,q_2 \in W$ and $q_3 \in W'$. If $i=0$, i.e., we are in the first step of the algorithm, then $q_1$ and $q_2$ are states of accepting SCCs.}%$q_1 \in S$ and $q_2 \in S'$, it holds that $S$ and $S'$ are accepting SCCs.}
\end{figure}

Algorithm~\ref{alg:counting} describes the procedure. An algorithm for the minimum model counting problem is defined in similar way. 
The algorithm works in a backwards fashion starting with states of accepting strongly connected components. In each iteration $i$, the algorithm maps each state of the automaton with $X \cup Y$-projected words of length $i$ that reach an accepting strongly connected component. After $|Q|$ iterations, the algorithm determines from the mapping of initial state $q_0$ a $Y$-projected word of length $|Q|$ with the maximum number of matching $X\cup Y$-projected words.

\begin{theorem}
	The decisional version of the maximum model counting problem over automata (MMCA), i.e. the question whether the maximum is greater than a given natural number $n$, is in $\mathit{NP}^{\#P}$.
\end{theorem}
\begin{proof}
	Let a B\"uchi automaton over an alphabet $2^\mathit{AP}$ for a set of atomic propositions $\mathit{AP}$ and sets $\AP_X,\AP_Y,\AP_Z \subseteq \AP$ and a natural number $n$ be given.
	We construct a nondeterministic Turing Machine $M$ with access to a $\#P$-oracle as follows: $M$ guesses a sequence $\sigma_Y \in 2^\mathit{AP_Y}$. It then queries the oracle, to compute a number $c$, such that $c = |\{\sigma_X \in (2^{\AP_X})^\omega \mid \exists \sigma_Z \in(2^{\AP_Z})^\omega.~ \sigma_X\cup\sigma_Y\cup\sigma_Z \in L(B)\}|$, which is a $\#P$ problem~\cite{DBLP:conf/lata/FinkbeinerT14}. It remains to check whether $n>c$. If so, $M$ accepts.
\end{proof}
The following theorem summarizes the main findings of this section, which establish, depending on the property, an exponentially or even doubly exponentially better algorithm (in the quantitative bound) over the existing model checking algorithm for HyperLTL.

\begin{theorem}
Given a Kripke structure $K$ and a quantitative hyperproperty $\varphi$ with bound $n$, the problem whether $K \models \varphi$ can be decided in logarithmic space in the quantitative bound $n$ and in polynomial space in the size of $K$.
\end{theorem}

%%%%%%%%%%%%%%%%%%%%
%\section{A SAT-based Approach for Model Checking Quantitative Hyperproperties}
\section{A Max\#Sat-based Approach}
For existential \hyperltl formulas $\psi_\iota$ and $\psi$, we give a more practical model checking approach by  encoding the automaton-based construction presented in Section~4 into a propositional formula.

 Given a Kripke structure $K=(S,s_0, \tau, \AP_K, L)$ and a quantitative hyperproperty $\varphi=\forall \pi_1, \dots, \pi_k.~ \psi_\iota \rightarrow (\CQuan \sigma :A.~ \psi) \triangleleft n $ over a set of atomic propositions $\AP_\varphi\subseteq\AP_K$ and bound $\mu$, our algorithm constructs a propositional formula $\phi$ such that,
 every satisfying assignment of $\phi$ uniquely encodes a tuple of lassos $(\pi_1, \dots, \pi_k,\sigma)$ of length $\mu$ in $K$, where $(\pi_1, \dots, \pi_k)$ satisfies $\psi_\iota$ and $(\pi_1, \dots, \pi_k,\sigma)$ satisfies $\psi$. 
 To compute the values $\max \limits_{(\pi_1,\dots,\pi_k)}|\{\sigma_A \mid (\pi_1,\dots,\pi_k,\sigma) \models \psi_\iota \wedge \psi \}|$ (in case $\triangleleft \in \{\leq, <\}$) or $\min \limits_{(\pi_1,\dots,\pi_k)}|\{\sigma_A \mid (\pi_1,\dots,\pi_k,\sigma) \models \psi_\iota \wedge \psi \}|$ (in case $\triangleleft \in \{\ge, >\}$),  we pass $\phi$ to a maximum model counter, respectively, to a minimum model counter with the appropriate sets of counting and maximization, respectively, minimization propositions.
  From Lemma~\ref{lem:maxnummodels} we know that it is enough to consider lassos of length linear in the product automaton. The size of $\phi$ is thus exponential in the size of $\varphi$ and polynomial in the size of $K$. 

The construction resembles the encoding of the bounded model checking approach for LTL~\cite{Clarke:2001:BMC:510986.510987}.
Let $\psi_\iota = \exists \pi'_1 \dots \pi'_{k'}.~\psi'_{\iota}$ and $\psi=\exists \pi''_1 \dots \pi''_{k''}.~\psi''$ and let $\AP_{\psi_\iota}$ and $\AP_{\psi}$ be the sets of atomic propositions that appear in $\psi_\iota$ and $\psi$ respectively.  
The propositional formula $\phi$ is given as a conjunction of the following propositional formulas:
$$\phi = \bigwedge_{i\leq k} \llbracket K \rrbracket ^{\pi_i}_\mu  \wedge \llbracket K \rrbracket ^\sigma_\mu \wedge \llbracket \psi_\iota \rrbracket _{\mu}^0 \wedge \llbracket \psi \rrbracket _{\mu}^0$$
where:
\begin{itemize}
	%%%%%%%%%%%         Mu         %%%%%%%%%%% 
	\item $\mu$ is  length of considered lassos and is equal to  $\mu = 2^{|\psi'_\iota \wedge \psi''|}*|S|^{k+k'+k''+1}+1$, which is one plus the size of the product automaton constructed from the $k+k'+k''+1$ self-composition and the automaton for $\psi_\iota \wedge \psi$. The "plus one" is to additionally to check whether the number of models is infinite.  

	%%%%%%%%%%%  Kripke Structure  %%%%%%%%%%%
	\item $\llbracket K \rrbracket ^\pi_k$ is the encoding of the transition relation of the copy of $K$ where atomic propositions are indexed with $\pi$ and up to an unrolling of length $k$. Each state of $K$ can be encoded as an evaluation of a vector of $\log{|S|}$ unique propositional variables. The encoding\footnote{We refer the reader to \cite{Clarke:2001:BMC:510986.510987} for more information on the encoding} is given by the propositional formula $ I(\vv{v}_0^\pi) \wedge \bigwedge_{i=0}^{k-1}\tau(\vv{v}_i^\pi,\vv{v}_{i+1}^\pi)$ which encodes all paths of $K$ of length $k$. 
	The formula $I(\vv{v}_0^\pi)$ defines the assignment of the initial state. The formulas $\tau(\vv{v}_i^\pi,\vv{v}_{i+1}^\pi)$ define valid transitions in $K$ from the $i$th to the $(i+1)$st state of a path. 
	%%%%%%%%%%%  HyperLTL formulas %%%%%%%%%%%
	\item $ \llbracket \psi_\iota \rrbracket _{k}^0$ and $\llbracket \psi \rrbracket _{k}^0$ are constructed using the following rules\footnote{We omitted the rules for boolean operators for the lack of space}:
		\begin{center}
			\small
			\begin{tabular}{|c||c|c|}
			\hline
			& $i<k$ & $i=k$ \\
			\hline
			$\llbracket a_\pi \rrbracket _k^i$ & $a^i_\pi$ & $\bigvee_{j=0}^{k-1} (l_j \wedge a_\pi^j)$ \\
			\hline 
			$\llbracket \neg a_\pi \rrbracket _k^i$ & $\neg a^i_\pi$& $\bigvee_{j=0}^{k-1} (l_j \wedge \neg a_\pi^j)$ \\
			\hline
			$\llbracket \LTLcircle \varphi_1 \rrbracket _k^i $ & $\llbracket \varphi_1 \rrbracket _k^{i+1} $ &  $\bigvee_{j=0}^{k-1} (l_j \wedge \llbracket \varphi_1 \rrbracket_k^j)$\\
			\hline
			$\llbracket \varphi_1 \LTLuntil \varphi_2 \rrbracket _k^i$ & $\llbracket \varphi_2 \rrbracket _k^i \vee (\llbracket \varphi_1 \rrbracket _k^i \wedge \llbracket \varphi_1 \LTLuntil \varphi \rrbracket _k^{i+1})$ & $\bigvee_{j=0}^{k-1} (l_j \wedge \langle \varphi_1 \LTLuntil \varphi_2 \rangle _k^j )$ \\
			\hline
			$\langle \varphi_1 \LTLuntil \varphi_2 \rangle _k^i$ & $\llbracket \varphi_2 \rrbracket _k^i \vee (\llbracket \varphi_1 \rrbracket_k^i \wedge \langle \varphi_1 \LTLuntil \varphi \rangle _k^{i+1})$ & false \\
			\hline
			$\llbracket \varphi_1 \LTLrelease \varphi_2 \rrbracket _k^i$ & $\llbracket \varphi_2 \rrbracket _k^i \wedge (\llbracket \varphi_1 \rrbracket _k^i \vee \llbracket \varphi_1 \LTLrelease \varphi \rrbracket _k^{i+1})$ & $\bigvee_{j=0}^{k-1} (l_j \wedge \langle \varphi_1 \LTLrelease \varphi_2 \rangle _k^j )$ \\
			\hline
			$\langle \varphi_1 \LTLrelease \varphi_2 \rangle _k^i$ & $\llbracket \varphi_2 \rrbracket _k^i \wedge (\llbracket \varphi_1 \rrbracket_k^i \vee \langle \varphi_1 \LTLrelease \varphi \rangle _k^{i+1})$ & true \\
			\hline
		\end{tabular}
		\end{center}
		in case of an existential quantifier over a trace variable $\pi$, we add a copy of the encoding of $K$ with new variables distinguished by $\pi$:
		\begin{center}
		\small
		\begin{tabular}{|c|c|}
			\hline 
			\quad $\llbracket \exists \pi.\varphi_1 \rrbracket _k^i $ \quad & \quad $\llbracket K \rrbracket^\pi_k  \wedge \llbracket \varphi_1 \rrbracket _k^i$ \quad \\
			\hline
		\end{tabular}
		\end{center} 
\end{itemize}

We define sets $X= \{a^i_\sigma \mid a \in A , i \leq k\}$, $Y= \{a^i \mid  a \in \AP_\psi \setminus A, i\leq k\}$	and $Z = P\setminus X \cup Y$, where $P$ is the set of all propositions in $\phi$. The maximum model counting problem is then $\textit{MMC}(\phi, X,Y,Z)$.

\subsection{Experiments}
We have implemented the Max\#Sat-based model checking approach from the last section. We compare the Max\#Sat-based approach to the expansion-based approach using \hyperltl~\cite{conf/cav/FinkbeinerRS15}. Our implementation uses the MaxCount tool~\cite{DBLP:conf/aaai/FremontRS17}. We use the option in MaxCount that enumerates, rather than approximates, the number of assignments for the counting variables. We furthermore instrumented the tool so that it terminates as soon as a sample is found that exceeds the given bound. If no sample is found after one hour, we report a timeout.  

Table~\ref{tab:comparison} shows the results on a parameterized benchmark obtained from the implementation of an 8bit passcode checker. The parameter of the benchmark is the bound on the number of bits that is leaked to an adversary, who might, for example, enter passcodes in a brute-force manner. In all instances, a violation is found.  The results show that the Max\#Sat-based approach scales significantly better than the expansion-based approach.  
\begin{table}
\scriptsize
\centering
\begin{tabular}{|c|c||c|c|c|c||c|c|c|}
\hline
	\multirow{2}{1.8cm}{\textbf{Benchmark}}& \multirow{2}{2cm}{\textbf{Specification}} & \multicolumn{4}{c||}{MCHyper} & \multicolumn{3}{c|}{MCQHyper}\\
	\cline{3-9}
	& &\#Latches & \#Gates & \#Quan. & Time(sec)  & \#max & \#count & Time(sec). \\
	\hline
	Pwd\_8bit & 1bit\_leak & \multirow{5}{1cm}{\center 9} & \multirow{5}{1cm}{\center 55} &  2 & 0.3 & 16 & 2 & 1\\
	\cline{1-2} \cline{5-9}
	& 2bit\_leak & & & 4 & 0.4 & 32 & 4 & 1\\
	\cline{1-2} \cline{5-9}
	& 3bit\_leak & & & 8 & 1.3 & 64 & 8 & 2\\
	\cline{1-2} \cline{5-9}
	& 4bit\_leak & & & 16 & 97  & 128 & 16 & 4 \\
	\cline{1-2} \cline{5-9}
	& 5bit\_leak & & &  32 & TO  & 256 & 32 & 8\\
	\cline{1-2} \cline{5-9}
	& 6bit\_leak & & & 64 &TO  & 512 & 64 & 335\\
	\cline{1-2} \cline{5-9}
	& 8bit\_leak & & &  256 &TO & 2048 & 256 & TO\\
	\hline
\end{tabular}
\caption{\scriptsize Comparison between the expansion-based approach (MCHyper) and the Max\#Sat-based approach (MCQHyper). \#max is the number of maximization variables (set $Y$). \#count is the number of the counting variables (set $X$). TO indicates a time-out after 1 hour.}
\label{tab:comparison}
\vskip -1.0cm
\end{table}

\section{Conclusion}
We have studied quantitative hyperproperties of the form $\forall \pi_1,\dots,\pi_k.\ \varphi \rightarrow (\# \sigma: A.\ \psi  \triangleleft n)$, where $\varphi$ and $\psi$ are HyperLTL formulas, and $\#\sigma:A. \varphi \triangleleft n$ compares the number of traces that differ in the atomic propositions $A$ and satisfy $\psi$ to a threshold $n$. Many quantitative information flow policies of practical interest, such as quantitative non-interference and deniability, belong to this class of properties.
Our new counting-based model checking algorithm for quantitative hyperproperties performs at least exponentially better in both time and space in the bound $n$ than a reduction to standard HyperLTL model checking. 
%While standard HyperLTL model checking performs as well as specialized algorithms for many hyperproperties of interest, such as non-interference or observational determinism, our results show that for quantitative hyperproperties this is not the case. 
The new counting operator makes the specifications exponentially more concise in the bound, and our model checking algorithm solves the concise specifications efficiently.

We also showed that the model checking problem for quantitative hyperproperties can be solved with a practical Max\#SAT-based algorithm. The SAT-based approach outperforms the expansion-based approach significantly for this class of properties.
An additional advantage of the new approach is that it can handle properties like deniability, which cannot be checked by MCHyper because of the quantifier alternation.

%%%% bib %%%%
%\newpage
\bibliographystyle{plain}
\bibliography{main}

%\newpage
%\appendix

\end{document}